\documentclass[12pt]{article}
\pdfoutput=1
\usepackage{amsmath}
\usepackage{graphicx,psfrag,epsf}
\usepackage{enumerate}
\usepackage{natbib}

\usepackage[ascii]{inputenc}
\usepackage{amsfonts}
\usepackage{amssymb}
\usepackage{color}
\usepackage{authblk}
\usepackage{setspace}
\usepackage{url}
\usepackage{verbatim}
\usepackage{epsfig}

\usepackage{amsthm}
\usepackage[small, margin=1cm]{caption}

\usepackage{amsmath,amssymb,amsthm,bm}
\usepackage{verbatim}
\usepackage{multirow}

\makeatletter
\@addtoreset{equation}{section}
\makeatother

\newtheorem{theorem}{\bf Theorem}[section]

\newtheorem{definition}[theorem]{Definition}
\newtheorem{example}[theorem]{Example}

\newtheorem{lemma}[theorem]{Lemma}

\newtheorem{proposition}[theorem]{Proposition}

\DeclareMathOperator*{\argmin}{arg\,min}

\theoremstyle{plain}

\theoremstyle{remark}

\newcommand{\polya}{P\'{o}lya}
\newcommand{\remove}[1] {}

\title{Computational Aspects of Optional \polya{} Tree}

\author[1,2,*]{Hui Jiang}
\author[3]{John C. Mu}
\author[4]{Kun Yang}
\author[2]{Chao Du}
\author[2]{Luo Lu}
\author[2,5,*]{Wing Hung Wong}
\affil[1]{Department of Biostatistics, University of Michigan}
\affil[2]{Department of Statistics, Stanford University}
\affil[3]{Department of Electrical Engineering, Stanford University}
\affil[4]{Institute for Computational and Mathematical Engineering, Stanford University}
\affil[5]{Department of Health Research and Policy, Stanford University}
\affil[*]{Please send all correspondence to jianghui@umich.edu and whwong@stanford.edu}
\begin{document}

\maketitle

\bigskip
\begin{abstract}Optional \polya{} Tree (OPT) is a flexible non-parametric Bayesian model for density estimation. Despite its merits, the computation for OPT inference is challenging. In this paper we present time complexity analysis for OPT inference and propose two algorithmic improvements. The first improvement, named Limited-Lookahead Optional \polya{} Tree (LL-OPT), aims at greatly accelerate the computation for OPT inference. The second improvement modifies the output of OPT or LL-OPT and produces a continuous piecewise linear density estimate. We demonstrate the performance of these two improvements using simulations.
\end{abstract}

\noindent%

\section{Introduction}

Given independent random samples, characterizing a distribution is one of the most fundamental tasks in statistics. Density estimation is a key approach to solve such problems and is closely related to other problems such as clustering, classification and regression~\citep{Silverman1986, Scott1992, Hastie2009}. Kernel density estimation~\citep{Rosenblatt1956, Parzen1962, Jones1996, Gray2003, Fan1994, Wand1994, Yang2003} and data-driven histograms~\citep{Scott1979, Barron1992, Lugosi1996, Klemela2009} are popular methods for non-parametric density estimation. Among many approaches developed for density estimation, non-parametric Bayesian models have the merits of being highly flexible and computationally feasible. Dirichlet process~\citep{Ferguson1973}, \polya{} tree~\citep{Ferguson1974} and their extentions are commonly used non-parametric Bayesian priors. Recently, \citet{WongLi2010} proposed an extension to the \polya{} tree prior, named Optional \polya{} Tree (OPT), which allows for optional stopping and random selection of partitioning variables. It has three attractive properties: 1) It fits the data adaptively; 2) It produces absolutely continuous distributions; and 3) The computation for exact OPT inference can be done in finite steps. OPT has been successfully applied in both one-sample~\citep{WongLi2010} and two-sample~\citep{LiWong2011} settings.

Albeit all these merits, there are existing issues which may hinder the application of OPT. Two major issues are: 1) The running time for exact OPT inference increases rapidly as sample size and the dimension of the data increase, which becomes prohibitive when the dimension is high; 2) OPT gives rise to piecewise constant densities\footnote{The posterior mean of OPT gives a smooth density. However, its estimation is usually computationally prohibitive~\citet{WongLi2010}. The hierarchical maximum a posteriori (hMAP) estimate proposed by~\citet{WongLi2010} is computationally more friendly, but gives rise to piecewise constant density estimates.}, where discontinuities occur at boundaries between the pieces, which could become undesirable for applications where smooth densities are sought. 

In this paper we attempt to tackle the above two problems. In particular, we demonstrate that the computation for OPT inference can be greatly accelerated by performing approximate inference. We also present an approach which modifies the output of OPT and produces a continuous piecewise linear density estimate. This paper is organized as follows: In Section 2 we review the construction of OPT and the algorithm for its exact inference. In Section 3 we present an approach for reducing the computation named Limited-Lookahead Optional \polya{} Tree (LL-OPT). Section 4 gives the time complexity analysis for OPT and LL-OPT. In Section 5 we present an approach for the estimation of smooth densities. Simulation studies are given in Section 6.

\section{Optional \polya{} Tree and its inference}
\label{sec:OPT}

Proposed by~\citet{WongLi2010}, OPT is an extension to the \polya{} tree prior~\citep{Ferguson1974}. OPT defines a procedure which gives rise to random probability measures on a space $(\Omega, \mu)$. In this paper we assume that $\Omega$ is a bounded hyperrectangle in $\mathbb{R}^p$ and $\mu$ is the Lebesgue measure. The theory of OPT applies to more general cases, e.g., $\mu$ is a counting measure when $\Omega$ is finite.

The OPT procedure operates simultaneously as two processes: the first process recursively and randomly partitions the space $\Omega$ and the second process recursively assigns a random probability measure into each region in the partition generated by the first process.

The recursive partitioning process works as follows\remove{(see Figure~\ref{OPT_construction})}: Starting from the root region $A=\Omega$, a random binary sample $S$ is drawn according to $\text{Bernoulli}(\rho(A))$ where $\rho(A)$ is a parameter of the OPT. If $S=1$, we stop further partitioning $A$ and mark it as a terminal region. If $S=0$, a random integer $J\in\{1,\ldots,J_A\}$ is drawn according to a probability vector $\boldsymbol\lambda(A)=(\lambda_1(A),\ldots,\lambda_{J_A}(A))$, where $J_A$ is the number of different ways of partitioning $A$ into subregions and $(\lambda_1(A),\ldots,\lambda_{J_A}(A))$ is a set of parameters of the OPT such that $\sum_{j=1}^{J_A}\lambda_j(A)=1$. If $J=j$, $A$ is partitioned in the $j$-th way into $I^j_A$ subregions $\{A^j_1,\ldots,A^j_{I^j_A}\}$ such that
$$A=\bigcup_{i=1}^{I^j_A}A_i^j.$$
The partitioning process is then recursively applied to each subregion $A_i^j$ of $A$, and consequently, a partition tree is built by the process recursively. The partitioning process terminates when all the regions in the partition are marked as terminal regions.

The process for random assignment of probability measures works along with the partitioning process\remove{ (see Figure~\ref{OPT_construction})}. It begins with a probability measure $Q^{(0)}$ which is uniform within the root region $A=\Omega$. Whenever a region $A$ is partitioned into subregions $\{A_i^j\}_{i=1}^{I^j_A}$, a random vector $\boldsymbol\theta^j_A=(\theta_{A,1}^j,\ldots,\theta_{A,I^j_A}^j)$ is drawn from a Dirichlet distribution with parameter $\boldsymbol{\alpha}^j(A)=(\alpha_1^j(A),\ldots,\alpha_{I^j_A}^j(A))$ and then a new probability measure $Q^{(k+1)}$ is built based on $Q^{(k)}$ by reassigning the probability mass within $A$ according to $Q^{(k+1)}(A_i^j)=\theta_{A,i}^jQ^{(k)}(A)$, were $k=k(A)$ is the level of region $A$ in the partition tree with $k(\Omega)=0$. Within each subregion $A_i^j$, the probability mass is again uniformly distributed.

For the OPT procedure described above, there exists a limiting probability measure for $Q^{(k)}$~\citep[Theorem 1]{WongLi2010}: If $\rho(A)$ is uniformly bounded away from $0$ and $1$ for all $A$, then $Q^{(k)}$ converges almost surely in variational distance to a probability measure $Q$ which is absolutely continuous with respect to $\mu$. The density of $Q$ with respect to $\mu$ on $\Omega$ is piecewise constant with countably many pieces.

\remove{
\begin{figure}[!htb]
\begin{centering}
\caption{The construction of an OPT distribution\label{OPT_construction}}
\end{centering}
\end{figure}
}

As a prior distribution on the space of probability measures, one merit of OPT is that the computation for exact OPT inference is analytically manageable. Given independent random samples $\boldsymbol{x}=(x_1,\ldots,x_n)$ from $Q$ when $Q$ has a OPT prior $\pi$ with independent stropping probabilities $\rho(A)$ and all the parameters of $\pi$ are uniformly bounded away from $0$ and $1$, the posterior distribution of $Q$ is also an OPT with the following parameters~\citep[Theorem 3]{WongLi2010}:
\begin{enumerate}
\item Stopping probabilities: $\rho(A|\boldsymbol{x})=\rho(A)\mu(\boldsymbol{x}|A)/\Phi(A),$
\item Selection probabilities:
$$\lambda_j(A|\boldsymbol{x})\propto\lambda_j(A)\frac{D(\boldsymbol{n}^j_A+\boldsymbol{\alpha}^j(A))}{D(\boldsymbol{\alpha}^j(A))}\prod_{i=1}^{I^j_A}\Phi(A_i^j), \mbox{\hspace{1cm}} j=1,\ldots,J_A,$$
\item Probability mass allocation:
$$\alpha^j(A|\boldsymbol{x})=\boldsymbol{\alpha}^j(A)+\boldsymbol{n}^j_A, \mbox{\hspace{1cm}} j=1,\ldots,J_A,$$
\end{enumerate}
where $\boldsymbol{n}^j_A=(n_{A,1}^j,\ldots,n_{A,I^j_A}^j)$ are the numbers of samples in $\boldsymbol{x}$ falling into each regions of the partition, $D(\boldsymbol{t})=\Gamma(t_1)\cdots\Gamma(t_k)/\Gamma(t_1+\cdots+t_k)$,  $\Phi(A)=P(\boldsymbol{x}(A)|A)$ is the marginal likelihood of the data conditional on $A$ which can be calculated using the following recursive equation:
\begin{equation}
\Phi(A)=\rho(A)\mu(\boldsymbol{x}|A)+(1-\rho(A))\sum_{j=1}^{J_A}\lambda_j(A)\frac{D(\boldsymbol{n}^j_A+\boldsymbol{\alpha}^j(A))}{D(\boldsymbol{\alpha}^j(A))}\prod_{i=1}^{I^j_A}\Phi(A_i^j).
\label{recursion}
\end{equation}

There are many different schemes for partitioning the regions. In this paper we consider a simple binary partitioning scheme proposed in~\citet{WongLi2010}: For $A=\{(t_1,\ldots\,t_p):t_j\in[l_j,u_j]\}$, i.e., a hyperrectangle in $\mathbb{R}^p$, there are $J_A=p$ ways to partition $A$. For the $j$-th way, $j=1,\ldots,p$, we partition $A$ into two subregions at the midpoint of the range of $t_j$ such that $A_1^j=\{t\in A:t_j< (l_j+u_j)/2\}$ and $A_2^j=A\backslash A_1^j$.

\section{Acceleration of computation}
\label{sec:algorithms}

According to~(\ref{recursion}), exact OPT inference requires computing $\Phi(A)$ recursively, which is computationally challenging since the running time increases rapidly as the dimension of the data increases. This is evident based on the fact that in~(\ref{recursion}) the value of $\Phi(A)$ has to be computed for every region $A$ with at least one data point and the number of such regions is enormous when the dimension is high. Figure~\ref{OPT_runtime} shows the running time for OPT inference for samples of various sizes and dimensions, using a fast algorithm for exact OPT inference which will be introduced Section~\ref{complexity}. We can see that the running time increases nearly exponentially with respect to the dimension of the data and it becomes prohibitive when the dimension is $8$ or higher. The analysis of the time complexity for OPT inference is given in Section~\ref{complexity}.

\begin{figure}[!htb]
\center
%\begin{tabular}{cc}
%\includegraphics[scale=0.4]{dim.pdf} & \includegraphics[scale=0.4]{num.pdf} \\ 
%(a) & (b) \\ 
%\end{tabular} 
%\includegraphics[scale=0.6]{compare.pdf}
\includegraphics[scale=0.6]{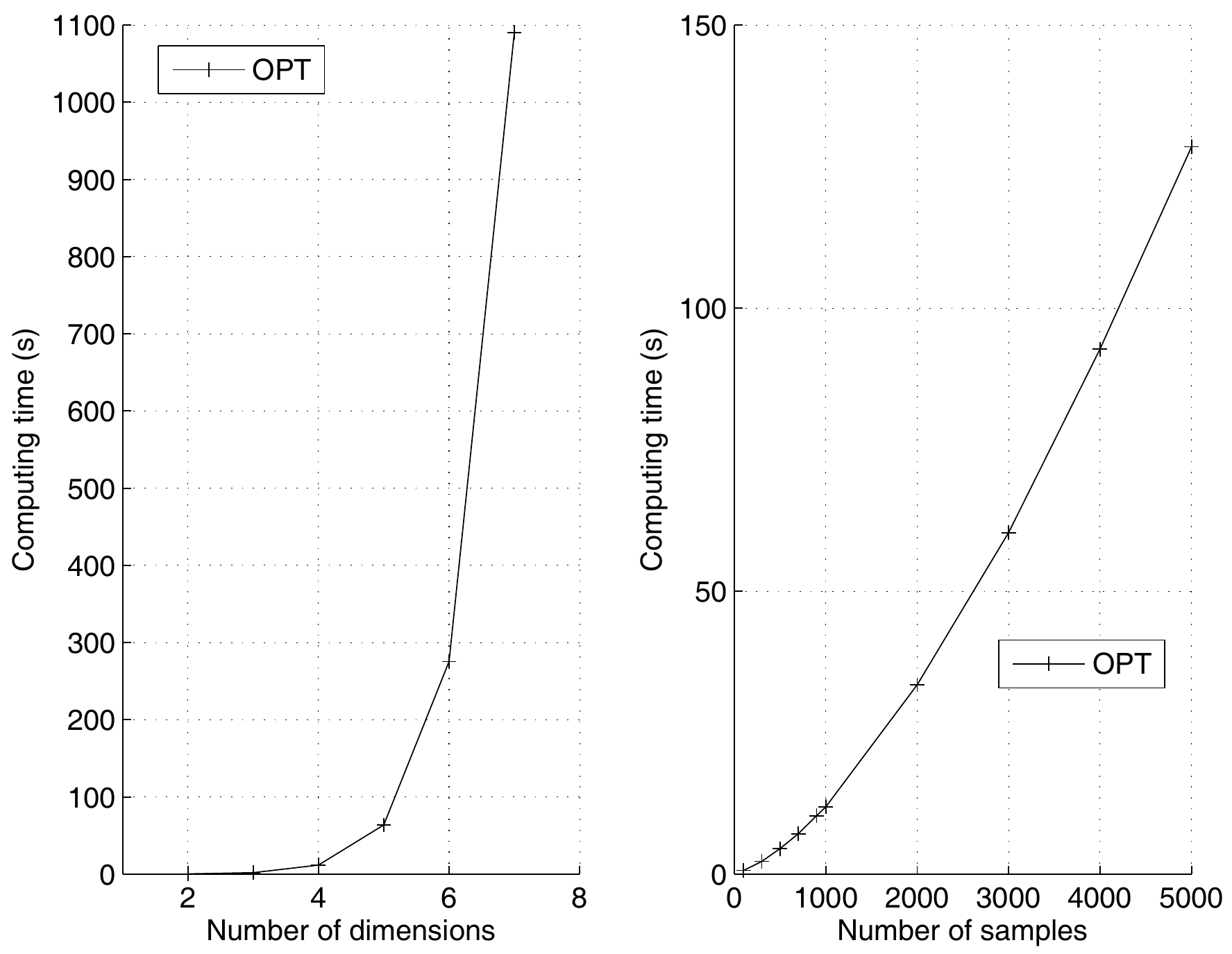}
\caption{\label{OPT_runtime} The running time for exact OPT inference using the cached implementation (See Section~\ref{complexity}). Running times (in seconds) are plotted against different dimensions with sample size fixed at 1000 (left figure) as well as different sample sizes with number of dimensions fixed at 4 (right figure). The samples are distributed as bivariate normal in two dimensions and uniform in other dimensions. \remove{The red lines are NI-OPT and the green lines are LL-OPT.}
}
\end{figure}

The computation for OPT inference can be greatly accelerated by performing the inference approximately. In a straightforward approach, named Naive Inexact Optional \polya{} Tree (NI-OPT), we stop the recursion in~(\ref{recursion}) when the number of samples in a region $A$ is small or when the volume of $A$ is small, and compute $\Phi(A)$ by assuming that all the samples in $A$ are uniformly distributed and therefore no further partitioning of $A$ is necessary. The rationale for employing these heuristics is that when the volume of $A$ is small, error in $\Phi(A)$ is only likely to affect the estimated density locally, and when the number of samples in $A$ is small, $\Phi(A)$ would not be too different from $\Phi^\prime(A)$ which is computed by redistributing the samples in $A$ uniformly. 

In practice, we found that the running time for the NI-OPT approach, although significantly reduced compared with the exact OPT approach, is still prohibitive when the sample size is large and when the dimension is high. In this paper we propose a more efficient algorithm for approximate OPT inference, named Limited-Lookahead Optional \polya{} Tree (LL-OPT), which works as follows\remove{ (see Figure~\ref{LL-OPT_alg})}: starting from the root region $A=\Omega$, we perform the recursive computation in~(\ref{recursion}) for up to $h$ levels in the partition tree, where $h$ is a tuning parameter which provides a trade-off between running time and estimation accuracy. For a region $B$ at level $k(B)=k(A)+h$, we stop the recursion and approximate $\Phi(B)$ as $\Phi^\prime(B)$ using the uniform approximation as we did in NI-OPT. Based on the computed $\Phi(A)$ and $\Phi(A_s)$ where $A_s (s\in S)$ are all the subregions of $A$, we compute the hierarchical maximum a posteriori (hMAP) tree using the approach described in~\citet{WongLi2010} for up to $q$ levels, where $q\leq h$ is another tuning parameter. In our implementation we fix $q=1$. For each leaf node in the hMAP tree, we recursively apply the LL-OPT algorithm and build a sub-tree by treating that leaf node as the root node. Furthermore, like in NI-OPT, we can also stop the procedure when the number of samples in the region is small enough, or when the volume of the region is small enough.

\remove{
\begin{figure}[!htb]
\begin{centering}
\caption{The Limited-Lookahead Optional \polya{} Tree (LL-OPT) algorithm\label{LL-OPT_alg}}
\end{centering}
\end{figure}
}

\section{Time complexity analysis for OPT and LL-OPT}
\label{complexity}

In this section we study the time complexity for the exact OPT and the LL-OPT algorithms. For simplicity we consider only the simple binary partitioning scheme described in Section~\ref{sec:OPT}. However, the analysis technique may be adapted to more complicated partitioning schemes. For each region $A$, the exact OPT algorithm according to~(\ref{recursion}) has to undertake the following three computational tasks:
\begin{enumerate}
\item For each subregion $A^j_i$ of $A$, $j=1,\ldots,p,i=1,2$, determine $n^j_{A,i}$, the counts of samples within $A^j_i$.
\item For each subregion $A^j_i$ of $A$, compute $\Phi(A^j_i)$ recursively using~(\ref{recursion}).
\item Calculate $\Phi(A)$ according to~(\ref{recursion}) based on the counts found in task 1 and $\Phi(A^j_i)$ computed for each subregion $A^j_i$ in task 2.
\end{enumerate}

For task 1, as $n\rightarrow\infty$, the straightforward way of counting the number of samples in each of the $2p$ subregions of $A$ will take $\Theta(2pN)$ operations, where $N$ is the total number of samples in $\Omega$. However, this cost can be reduced by trading space for time. If we can store the set of samples in $A$, counting the number of samples in each of the $2p$ subregions of $A$ will only take $\Theta(2pn)$ operations, where $n$ is the number of samples in $A$.\remove{ To keep a manageable memory footprint. In our implementation the partition tree is searched in a depth-first manner so that only the samples for the regions currently stored in the call stack need to be stored.} Let $f(n)$ denote the time for computing $\Phi(A)$ for a region $A$ with $n$ samples. The time complexity for task 2 is therefore $\sum_{j=1}^p\{f(n^j_{A,1})+f(n^j_{A,2})\}$. For task 3, the time complexity is $\Theta(2p)$ if we assume that function $D(\mathbf{t})$ can be computed in $\Theta(1)$ time. The time complexity for computing $\Phi(A)$ is therefore
\begin{equation}
\label{f_n}
\begin{array}{lll}
f(n)&=&\Theta(2pn)+\sum_{j=1}^p\{f(n^j_{A,1})+f(n^j_{A,2})\}+\Theta(2p)\\
	&=&\Theta(pn)+\sum_{j=1}^p\{f(n^j_{A,1})+f(n^j_{A,2})\}
\end{array}
\end{equation}

It is shown in~(\ref{f_n}) that the value of $f(n)$ depends on the values of $n^j_{A,1}$ and $n^j_{A,2}$, which are determined by the distribution of samples within $A$. The more uniformly the samples are distributed in $A$, the less the computation time is required. Here we estimate $f(n)$ for the following three cases:
\begin{enumerate}
\item $n^j_{A,1}=tn$ and $n^j_{A,2}=(1-t)n$, where $0.5\leq t<1$ is a constant.
\item $n^j_{A,1}=Tn$ and $n^j_{A,2}=(1-T)n$, where $T$ is uniformly distributed within $(0,1)$. 
\item $n^j_{A,1}=c$ and $n^j_{A,2}=n-c$, where $c>0$ is a constant.
\end{enumerate}
We see that case 2 represents a more realistic picture between the two extremes in case 1, namely $t=0.5$ and $t$ is close to $1$, and case 3 is the limit of case 1 as $1-t$ goes to $c/n$. In effect, case 1 indicates that the number of levels to be explored is $\log_{1/t}n$, and case 3 indicates that the number of level is $n/c$. Although only 3 specific cases are discussed here, the results actually provide full coverage over general scenarios.

\begin{theorem}
\label{thm_OPT}
For cases 1, 2 and 3, the time complexity for the exact OPT algorithm is
$f(n)=O(pn^{\log_{1/t}2p})$ for case 1, $f(n)=O(n^{2p-1})$ for case 2 and $f(n)=O(np^{n/c})$ for case 3, where $n$ is the sample size and $p$ is the dimension.
\end{theorem}

\begin{proof}[Proof of Theorem~\ref{thm_OPT}]

For case 1, from~(\ref{f_n}) we assume $f(n)\approx C_1pn+p\{f(tn)+f((1-t)n)\}$ for large $n$, where $C_1$ is a constant. \citet{Akra-Bazzi} showed that the solution to the above linear recurrence equation is
$$f(n)=\Theta\left\{n^q\left(1+\int_1^n\frac{C_1pu}{u^{q+1}}du\right)\right\}$$
where $q$ is the solution to the equation $pt^q+p(1-t)^q=1$. Since usually $t$ is unknown, here we derive a upper bound for $f(n)$ as follows.
\begin{eqnarray*}
f(n)&\leq&C_1pn+2pf(tn)\\
	&\leq&C_1pn+2ptC_1pn+(2pt)^2C_1pn+\cdots+(2pt)^{\log_{1/t}n}C_1pn\\
	&=&O(pn(2pt)^{\log_{1/t}n})\\
	&=&O(pn^{\log_{1/t}2p}).
\end{eqnarray*}
In the best case $t=0.5$ and $f(n)=O(pn^{1+\log_2p})$.

Similarly, we can get a lower bound, 
\begin{eqnarray*}
f(n)&\geq&C_1pn+pf(tn)\\
	&\geq&C_1pn+ptC_1pn+(pt)^2C_1pn+\cdots+(pt)^{\log_{1/t}n}C_1pn\\
	&=&\Omega(pn(pt)^{\log_{1/t}n})\\
	&=&\Omega(pn^{\log_{1/t}p}).
\end{eqnarray*}
In the best case $t=0.5$ and $f(n)=\Omega(pn^{\log_2p})$.

For case 2, 
\begin{eqnarray*}
f(n)&=&C_1pn+p\int_0^1\{f(tn)+f(n-tn)\}dt\\
	&=&C_1pn+2p\int_0^1f(tn)dt.\\
\end{eqnarray*}
Regard $n$ as a real-valued parameter and take derivative over $n$ on both sides,
\begin{eqnarray*}
f^\prime(n)&=&C_1p+2p\int_0^1tf'(tn)dt\\
	 &=&C_1p+2p\left\{\left[\frac{t}{n}f(tn)\right]_0^1-\int_0^1\frac{f(tn)}{n}dt\right\}\\
	 &=&2C_1p+\frac{(2p-1)f(n)}{n}.
\end{eqnarray*}
The solution to the above ODE is $f(n)=\frac{C_1pn}{1-p}+C_2n^{2p-1}=O(n^{2p-1})$.

For case 3, 
\begin{eqnarray*}
f(n)&=&C_1pn+p\left\{f(c)+f(n-c)\right\}\\
	&=&p\{C_1n+f(c)\}+p^2\{C_1n+f(c)\}+\cdots+p^{n/c}\{C_1n+f(c)\}\\
	&=&O(np^{n/c})
\end{eqnarray*}
In the worst case $c=1$ and $f(n)=O(np^n)$.
\end{proof}

The simplest implementation of the OPT algorithm is to use depth-first search, which we call DF-OPT. We find DF-OPT to be very inefficient in practice, partly because regions that can be reached through many different paths of binary partitions will have their $\Phi(\cdot)$ values computed repeatedly, which is unnecessary. A more efficient implementation is to cache all the computed $\Phi(\cdot)$ values so that for each region $A$, $\Phi(A)$ will need to be computed only once. In practice we find that this cached implementation can accelerate computation substantially. Therefore, we use the cached implementation as the default implementation for all the simulations (denoted as OPT in experiments). Comparisons between DF-OPT and cached OPT are given in Section~\ref{sec:simulations}. It is worthwhile pointing out here that although cached OPT is much faster than DF-OPT, it only works with limited sample sizes and dimensions because it uses much larger amount of memory. In contrast, DF-OPT has the advantage that its memory usage is minimal. Furthermore, if a partitioning scheme other than the binary partitioning is used, caching might become less effective or even non-effective.

We found that in practice, the cached implementation of OPT is most likely to be limited by memory size rather than by running time. Therefore, we also study the space complexity of the cached OPT here.

\begin{lemma}
\label{lem_region_count}
Under the binary partitioning scheme, the total number of regions at level $k$ is ${k+p-1\choose k}2^k$, where $p$ is the dimension of the sample space. The total number of regions up to level $k$ is therefore $\sum_{i=0}^k{i+p-1\choose i}2^i$.
\end{lemma}

\begin{proof}
Without loss of generality, assume the root region $\Omega$ is the unit hypercube $[0,1]^p$ in $\mathbb{R}^p$. Under the binary partitioning scheme, any subregion of $\Omega$ can be uniquely encoded by a string consisting of four symbols: `$\epsilon$', `$\mid$', `0' and `1'. For instance, the level-3 subregion $[0,1]\times[.25,.5]\times[.5,1]$ can be encoded by string `$\epsilon\mid01\mid1$', with `$\epsilon$' encodes the interval $[0,1]$, `01' encodes the interval $[.25,.5]$ and `1' encodes the interval $[.5,1]$. Specifically, the binary string $s=`a_1a_2\ldots a_k\textrm'$ encodes the interval $[s/2^k,(s+1)/2^k]$ in binary.
We can therefore count the number of level-$k$ regions by enumerating all the valid region-encoding strings: start from all the $2^k$ binary sequences of length $k$, insert $p-1$ `$\mid$'s in the strings to break it into $p$ segments and insert `$\epsilon$'s in all the empty segments. Because for each binary string of length $k$, there are ${k+p-1\choose k}$ ways of inserting $p-1$ `$\mid$'s into it, the total number of regions is therefore $k$ is ${k+p-1\choose k}2^k$.
\end{proof}

If we assume that the depth of the partition tree is on the order of $\log_2n$ where $n$ is the sample size, which is the case when the samples in $\Omega$ are distributed roughly uniformly, the space required by the cached OPT is $\Theta\{\sum_{i=0}^{\log_2n}{i+p-1\choose i}2^i\}=\Theta\{n{\log_2n+p-1\choose \log_2n}\}$.

The LL-OPT algorithm is more friendly in CPU and memory requirements because for each cut it only looks ahead for $h$ steps, and as a result, the time and space complexity for a single cut is bounded.

\begin{lemma}
\label{lem_count}
Let $A^{(k)}_1,\ldots,A^{(k)}_{m^{(k)}}$ be all the child regions of $\Omega$ at level $k$. Let $n(A)$ be the number of samples in region $A$. Then $\sum_{i=1}^{m^{(k)}}n(A^{(k)}_i)\leq np^k$ where $p$ is the dimension of the space.
\end{lemma}

\begin{proof}
We prove the lemma by induction. For $k=0$, $\Omega$ is the only level-$0$ region, therefore $m(0)=1$ and $\sum_{i=1}^{m^{(0)}}n(A^{(0)}_i)=n=np^0$. Suppose the lemma holds for $0\leq j\leq k$. All the level-$(k+1)$ regions are generated by partitioning some level-$k$ regions. For each level-$k$ region $A$, the sum of sample sizes of its $2p$ child regions is $pn(A)$. Therefore, $\sum_{i=1}^{m^{(k+1)}}n(A^{(k+1)}_i)\leq p\sum_{i=1}^{m^{(k)}}n(A^{(k)}_i)\leq pnp^k= np^{k+1}$.
\end{proof}

\begin{theorem}
\label{thm_LL-OPT}
For the LL-OPT algorithm with parameter $h$, the time complexity is $f(n)=O\{n(n+2^hp)p^h\}$, where $n$ is the sample size and $p$ is the dimension.
\end{theorem}

\begin{proof}[Proof of Theorem~\ref{thm_LL-OPT}]
For the LL-OPT algorithm, for each region $A$, we only need to perform the recursive computation in~(\ref{recursion}) for up to $h$ levels deep. For each child region $B$ of $A$, the time for computing $\Phi(B)$ given the $\Phi$ values of all $B$'s subregions is bounded by $O(p)$. The time for computing $\Phi(A)$ is therefore bounded by $O(pK_h)$, where $K_h$ is the total number of child regions of $A$ up to $h$ levels deep, which is bounded by $O(2p)^h$. According to Lemma~\ref{lem_count}, the time for counting the number of samples in child regions up to $h$ levels deep is bounded by $O(\sum_{i=0}^hnp^i)=O(np^h)$. The time complexity for expanding one node of the tree is therefore $O\{p(2p)^h+np^h\}=O\{(n+2^hp)p^h\}$. Since the final hMAP tree built by the LL-OPT algorithm has $O(n)$ nodes, the time complexity of the LL-OPT algorithm is $O\{n(n+2^hp)p^h\}$.
\end{proof}

Theorems~\ref{thm_OPT} and~\ref{thm_LL-OPT} show that the reduction on running time using the LL-OPT algorithm can be significant, especially when $n$ and $p$ are large. The amount of reduction is controlled by the tuning parameter $h$. The smaller the $h$, the faster the LL-OPT algorithm. However, a smaller $h$ also leads to less accurate estimation. The following adaptive approach can be used for selecting $h$: run the LL-OPT algorithm many times starting with $h=1$ and increase $h$ by 1 each time until a stopping decision is made (either based on a cost constraint or a criterion evaluating the gain between two successive estimates). Theorem~\ref{thm_LL-OPT} shows that the total time for running the LL-OPT algorithm for $h=1,\ldots,k-1$ is less than the run time for a single $h=k$.

\section{Estimation of smooth density}

The densities estimated by OPT are piecewise constant and have discontinuities at boundaries between the pieces, which could become undesirable for applications where smooth densities are sought.\remove{Two approaches were suggested in~\citet{WongLi2010} for density estimation using OPT: (1) Compute the posterior mean density; (2) First learn a fixed tree topology that is representative of the underlying structure of the distribution, and then compute a piecewise constant estimate conditioned on the tree topology. Approach (2) is more attractive because its computational ease. However, the discontinuity of the density estimated by using the tree topology will hinder the application where continuity is required.  Furthermore, the piecewise constant estimates tend to flatten the peaks and valleys in the true density landscape, which deteriorate the estimation accuracy.} In this section we introduce a procedure which constructs a continuous piecewise linear approximation to the piecewise constant density estimated by OPT via quadratic programming. We call the continuous piecewise linear density estimator constructed using this approach the Finite Element Estimator (FEE), because it was inspired by the idea of domain partitioning in Finite Element Method (FEM)~\citep{Allaire}.

We define a proper space for continuous piecewise linear densities as follows. We start with the definition of a triangulation of the domain $\Omega$ (without loss of generality, we shall assume $\Omega$ is $[0, 1]^p$) by $p$-simplices. A $p$-simplex $\Delta$ in $\mathbb{R}^{p}$ is the convex envelop of its $p + 1$ vertices  $\{a_j\}_{1\leq j\leq p + 1}$, which is nondegenerate if the its vertices do not fall on a hyperplane in $\mathbb{R}^p$ and consequently the volume of $\Delta$, denoted as $\mu(\Delta)$, is not zero, which we shall always assume in what follows.
 \begin{definition}
   Let $\Omega = [0, 1]^p$, a triangulation of $\Omega$ is a set $\Gamma$ of $p$-simplices $\{\Delta_i\}_{1\leq i\leq m}$ satisfying
   \begin{itemize}
     \item[(i)] $\Delta_i\subset\Omega$ $(i=1\ldots m)$ and $\Omega=\cup_{i=1}^m\Delta_i$
     \item[(ii)] For any two distinct simplices $\Delta_i$ and $\Delta_j$, if they share $k~(1\leq k\leq p)$ vertices, $\Delta_i\cap \Delta_j$ is a $(k-1)$-simplex formed by these $k$ vertices, or an empty set otherwise. This condition guarantees that no vertex of $\Delta_i$ lies on the faces of $\Delta_j$ and vice versa. This can help exclude the undesired situations as shown in Figure~\ref{undesire}.
   \end{itemize}
 \end{definition}
 \begin{figure}[!htb]
   \center
   \includegraphics[scale=.3]{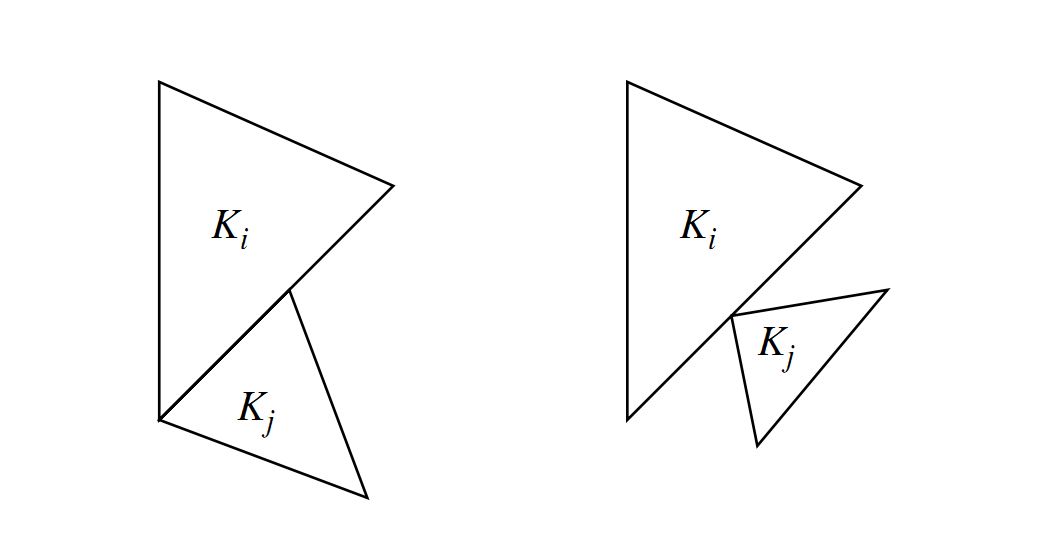}
   \caption{Examples of undesired situations for a triangulation.\label{undesire}}
 \end{figure}

\begin{definition}
For an OPT partition $\Omega=\bigcup_{i=1}^lA_i$, a triangulation of the OPT partition is a triangulation $\Gamma$ of $\Omega$ where each region $A_i$ is also triangulated by a subset of $\Gamma$. See Figure \ref{trian} for an example.
\end{definition}
 \begin{figure}[!htb]
  \center
  \includegraphics[scale=.8]{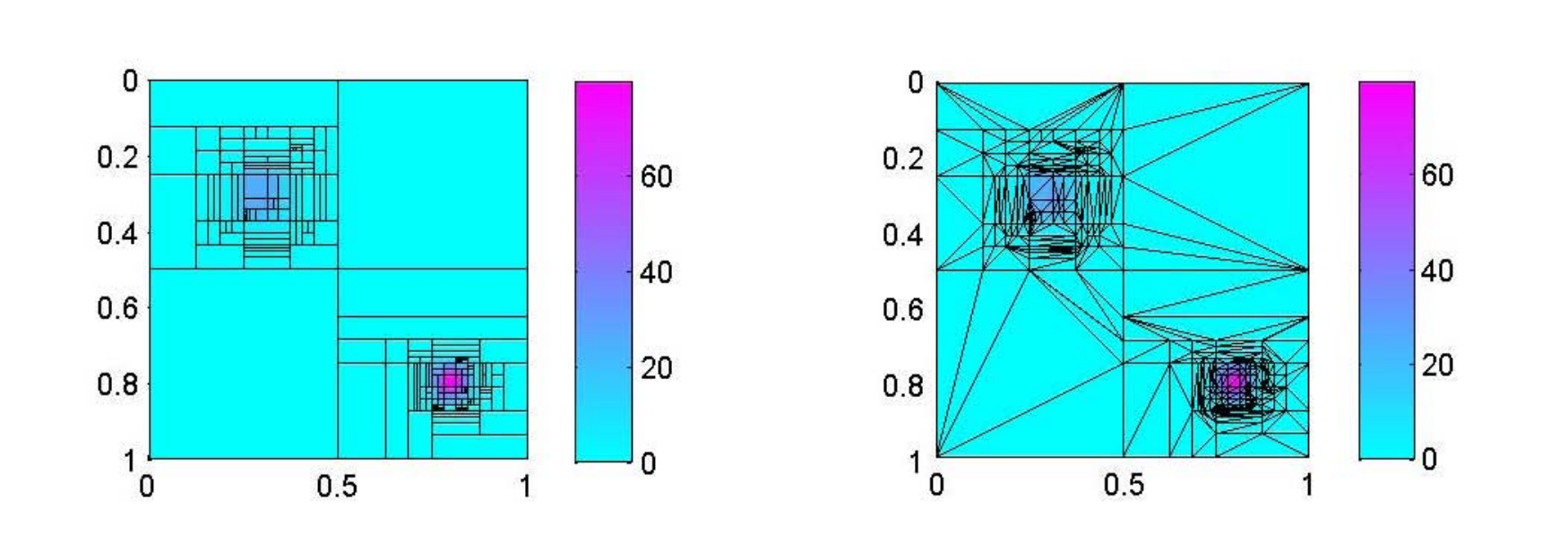}
   \caption{A OPT partition in 2-D (left) and its triangulation (right).\label{trian}}
\end{figure}

\begin{definition}
Given a triangulation $\Gamma$ of an OPT partition of $\Omega$, a function $f:\Omega\rightarrow\mathbb{R}$ is called continuous piecewise linear over $\Gamma$ if $f$ is linear within each simplex in $\Gamma$. The set of continuous piecewise linear functions over $\Gamma$ is denoted as  $\mathcal{P}_1(\Gamma)$.
\end{definition}

Let the set $\{a_1\ldots a_k\}$ be all the vertices in $\Gamma$, we define $k$ continuous piecewise linear functions as $\phi_i(a_j) = \delta_{ij}$ for all $1\leq i, j\leq k$. In another word, $\phi_i$ is $1$ at $a_i$ and $0$ at all other vertices. The set of functions $\phi_i~(i=1\ldots k)$ forms a basis for all continuous piecewise linear functions over $\Gamma$~\citep{Allaire}. One such basis function is shown in~Figure~\ref{basis}. 
\begin{figure}[!htb]
  \center
  \includegraphics[scale=.8]{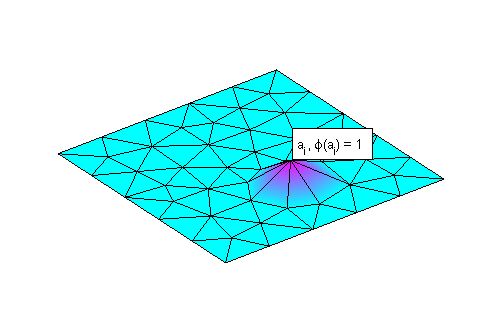}
  \caption{A continuous piecewise linear basis function.\label{basis}}
\end{figure}

Since the piecewise constant density estimated by OPT is a good approximation to the underlying distribution when the sample size is large, guaranteed by the consistency of OPT estimates~\citep[Theorem 4]{WongLi2010}, we aim at achieving the following two properties with the continuous piecewise linear FEE.
\begin{enumerate}
  \item (fidelity property) The difference between the piecewise constant density and the continuous piecewise linear density must be small.
  \item (smoothness property) The variation of the continuous piecewise linear density within each region of the OPT partition must be small.
\end{enumerate}

Unfortunately, in practice, these two properties often contradict each other. Therefore, we use penalized optimization to achieve a trade-off. Let $f\in\mathcal{P}_1(\Gamma)$ be a density function which is continuous piecewise linear over $\Gamma$, the fidelity property can be enforced by the following penalty $p(\cdot)$
$$p(f) = \sum_{j = 1}^mw_j\left[\int_{\Delta_j}f(x)dx - Q(\Delta_j)\right]^2,$$
where $\Delta_j$ is the $j$-th simplex in $\Gamma$, $Q(\Delta_j)$ is the probability mass assigned in $\Delta_j$ by the piecewise constant OPT density estimate and $w_j$ is a pre-specified weight depending on the importance of simplex $j$. In our simulations, we find that choosing $w_j$ to be $\mu(\Delta_j)^{-2}$ provides good results by keeping a balance between large and small simplices.

The smoothness property implies that the density function in any given simplex should be as flat as possible. Therefore, in a density plot (see Figure~\ref{faces} for an example), the volume of the simplex on the top, denoted as $\mu^*(\Delta_j)$, should be approximately equal to the volume of the simplex on the bottom, which is $\mu(\Delta_j)$. Therefore, the smoothness property can be enforced by the following penalty $q(\cdot)$
$$ q(f) = \sum_{i=1}^m\left[\frac{\mu^*(\Delta_j)}{\mu(\Delta_j)}\right]^2.$$
\begin{figure}[!htb]
  \center
  \includegraphics[scale=0.5]{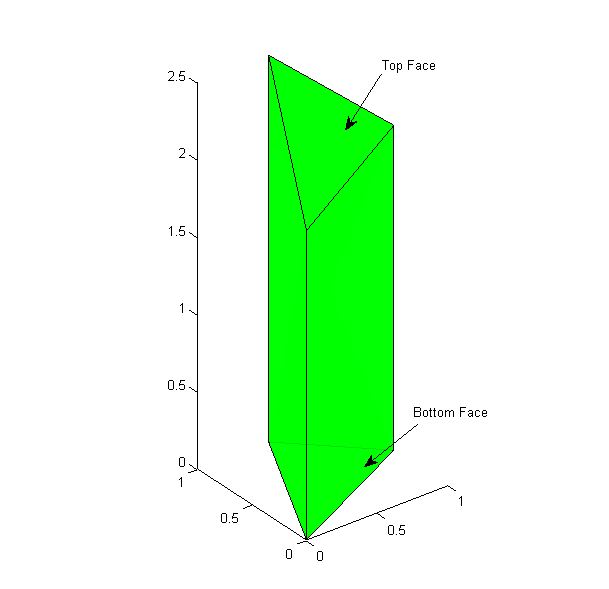}
  \caption{The bottom face (the simplex $\Delta_j$) and top face in a density plot.\label{faces}}
\end{figure}

To enforce both properties simultaneously, we minimize $h(f) = p(f) + \lambda q(f)$, where $\lambda$ is a tuning parameter. Therefore, the continuous piecewise linear density estimate $f$ can be obtained by solving the following optimization problem.
\begin{equation}
\label{eq2}
\begin{split}
&f=\argmin_{f\in\mathcal{P}_1(\Gamma)}\sum_{j = 1}^mw_j\left[\int_{\Delta_j}f(x)dx - Q(\Delta_j)\right]^2 + \lambda\sum_{i=1}^m\left[\frac{\mu^*(\Delta_j)}{\mu(\Delta_j)}\right]^2,\\
&\textrm{subject to }   f(x) \geq 0 \textrm{ for all } x\in\Omega \textrm{ and }   \sum_{i = 1}^m \int_{\Delta_i}f(x)dx = \int_{\Omega}f(x)dx = 1.
\end{split}
\end{equation}
Within a given simplex $\Delta_j$ with vertices $\{a_1, \ldots, a_{p + 1}\}$, and $\{c_1, \ldots, c_{p + 1}\}$ = $\{f(a_1),\ldots,f(a_{p+1})\}$ are the corresponding densities on the vertices, $f(x)$ can be written as a linear combination of the basis functions
$$  f(x) = \sum_{i = 1}^{p+1} c_i\phi_i(x),$$
where $c_i \geq 0$. We have
\[\mu(\Delta_j) = \frac{1}{p!}\left|\det\left(\begin{array}{ccccc}
a_{1, 1} & a_{1, 2} & \ldots& a_{1, p} & 1\\
a_{2, 1} & a_{2, 2} & \ldots & a_{2, p} & 1\\
&&\ldots&&\\
a_{p + 1, 1} & a_{p + 1, 2} & \ldots & a_{p + 1, p} & 1\\
\end{array}\right)\right|,\]
\[\int_{\Delta_j}f(x)dx = \frac{\mu(\Delta_j)}{p + 1}(c_1 + \cdots + c_{p + 1}).\]
and
$$\mu^*(\Delta_j) = \frac{1}{d!}\left\|\det\left(\begin{array}{ccccc}
  \mathbf{i}_1 & \ldots & \mathbf{i}_p & \mathbf{i}_{p + 1} \\
  a_{2,1}-a_{1,1}  & \ldots & a_{2,p}-a_{1,p} & c_2 - c_1 \\
  &&\ldots&&\\
  a_{p+1,1}-a_{1,1} & \ldots & a_{p+1,p}-a_{1,p} & c_{p+1} - c_1 \\
\end{array}\right)\right\|_2
$$
where $\|\cdot\|_2$ is the Euclidean norm and $\mathbf{i}_j$~$(j = 1,\ldots, p + 1)$ are the natural basis of $\mathbb{R}^{p+1}$. Since $\int_{\Delta_j}f(x)dx$ is a linear function of $c_1, \ldots, c_n$ and $\mu^*(\Delta_j)^2$ is a quadratic function of $c_1, \ldots, c_n$,~(\ref{eq2}) can be solved by quadratic programming.

\section{Simulations}
\label{sec:simulations}

\subsection{LL-OPT}

We compare the running times and the estimation accuracies of the exact OPT and the LL-OPT algorithms. In our simulations, the root region $\Omega$ is the unit square $[0,1]\times [0,1]$ on $\mathbb{R}^2$. The exact OPT algorithm and the LL-OPT algorithm are used to estimate the density based on the simulated samples with different sizes. The Hellinger distance $H(f,g)=\{\frac{1}{2} \int (f(x)^{1/2}-g(x)^{1/2})^2dx\}^{1/2}$ was used as the metric to evaluate the accuracy of density 
estimates.

\begin{example}
\label{beta_mixture_example}
We consider a mixture distribution of two independent components over the unit square $[0,1]\times [0,1]$~\citep[Example 8]{WongLi2010}. The first component is a uniform distribution over $[0.78, 0.80]\times[0.2, 0.8]$. The second component has support $[0.25, 0.4]\times[0, 1]$ with $X$ being uniform over $[0.25,0.4]$ and $Y$ being $Beta(100, 120)$. The density function is therefore
\begin{equation}
\label{semibeta} 
\frac{0.35}{0.012} \times \boldsymbol{1}_{[0.78, 0.80]\times[0.2, 0.8]}+\frac{0.65}{0.15}\times \frac{\Gamma(220)}{\Gamma(120)\Gamma(100)} y^{99} (1-y)^{119} \boldsymbol{1}_{[0.25,0.4]\times[0, 1]}.
\end{equation}
\end{example}

Figure~\ref{semibeta1000} shows the partition plots based on the exact OPT and the LL-OPT algorithms with $h=1,\ldots,5$, respectively, for a sample size of 1000. The plots show how the root region $\Omega$ is partitioned according to the estimated tree topologies. A good partition plot should visually resemble a contour plot of the distribution to be estimated. On these partition plots, the two components of the mixture are clearly marked out by the two vertical strips located on the left-half and the right-half of the square, respectively. Furthermore, the non-uniformity of the Beta distribution in the left component is also visible through the further division of the corresponding strip.  

\begin{figure}[!htb]
\begin{center}
\includegraphics[scale=0.6]{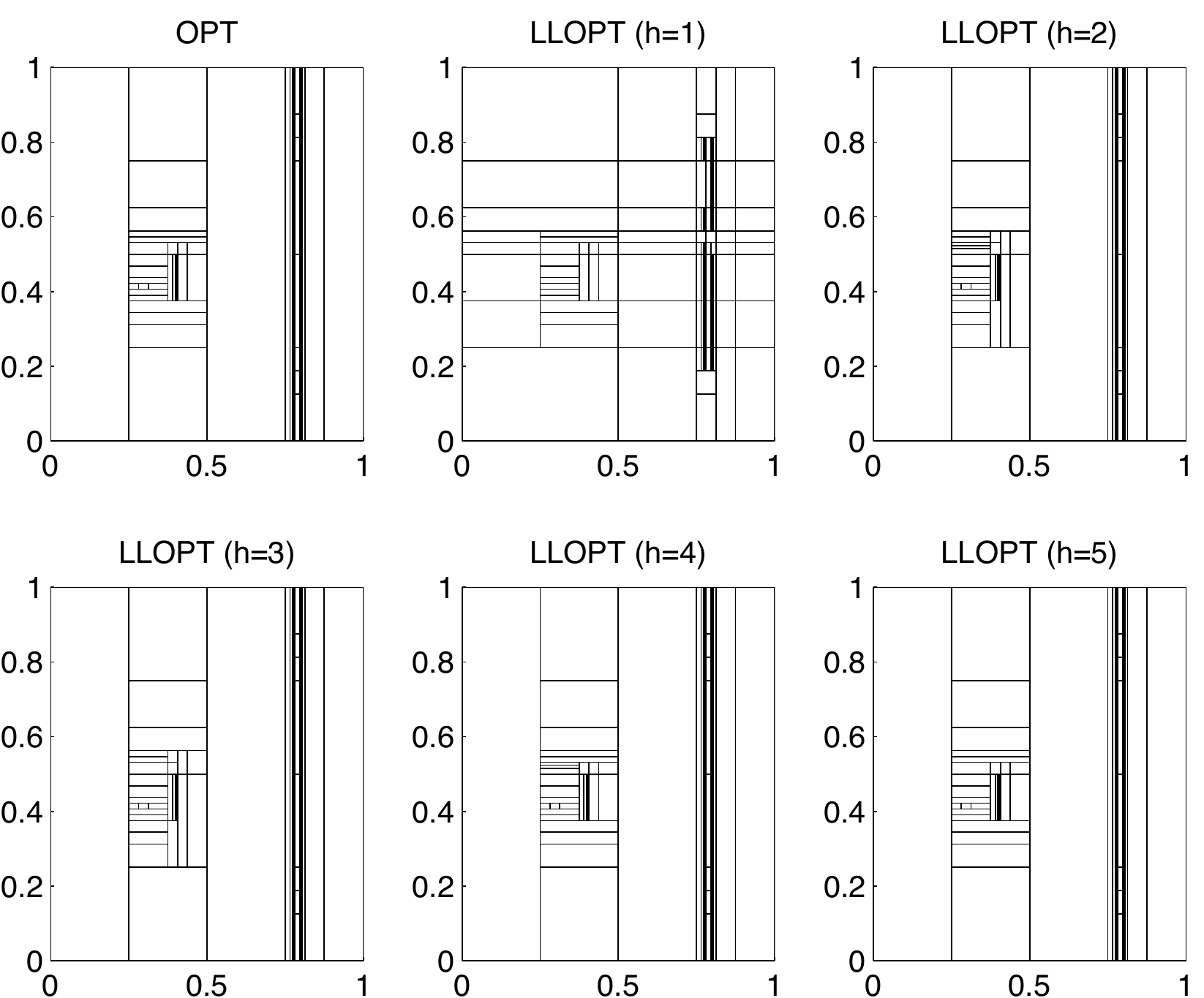}
\caption{The partition plots based on the exact OPT and the LL-OPT algorithms with $h=1,\ldots,5$ for Example~\ref{beta_mixture_example}. The samples are drawn from a mixture of uniform and ``semi-beta'' distributions as defined by~(\ref{semibeta}). The sample size is 1000.}
\label{semibeta1000}
\end{center}
\end{figure}

Comparing the partition plots between the exact OPT and the LL-OPT algorithms, the partition estimated by the LL-OPT algorithm, with $h$ as small as 2, resembles closely the partition estimated by the exact OPT algorithm with only minor differences. As $h$ increases, the resemblance becomes even stronger, which is expected as the two algorithms should lead to exactly the same result when $h$ is greater than the maximum depth. 

\remove{There exist noticeable differences between the partitions estimated by the exact OPT algorithm and the LL-OPT algorithm for $h=1$. The exact OPT algorithm partitions the root region at $x=1/2$, while the tree topology suggested by the LL-OPT algorithm partitions the root region at $y=1/2$. Generally speaking, algorithms based on the framework of optional P{\'o}lya tree prefer the partitioning scheme maximizing the difference between subregions. As the two components of the mixture distribution are separately located in the left and right part of the unit square~\ref{semibeta}, it is expected that root region should be split at $x=1/2$. However, in the LL-OPT algorithm with $h=1$, only  level-$1$ regions are considered for the purpose of determining the partitioning scheme of the root region. With respect to the true density, the probability masses in the left and right half are 0.65 and 0.35, respectively, while the masses are roughly 0.77 and 0.23 for the lower and upper half. Thus, rather than partitioning the root region at $x=1/2$, the fast algorithm with $h=1$ partitions the root region at $y=1/2$. This issue disappears as long as the deeper structures are taken into consideration ($h>1$).}

The average running times used by the exact OPT algorithm and the LL-OPT algorithm with different $h$ are summarized in Table~\ref{table_semibeta_cost}, and the average Hellinger distances (computed based $2\times10^6$ importance samples) between the estimated and the true densities are summarized in Table~\ref{table_semibeta_dis}. All the summary statistics are based on 5 replications and corresponding standard deviations are included in the parentheses.

\begin{table}[!htb]
\begin{center}
\caption{Running times (in seconds) for Example~\ref{beta_mixture_example}. Average of 5 replicates, standard deviation reported in parentheses. The DF-OPT algorithm could not finish the two large-sample simulations within reasonable time. \label{table_semibeta_cost}}
\begin{tabular}{cccccccc} 
\hline
\multirow{2}{1.5cm}{Sample Size} & \multirow{2}{1.5cm}{OPT} & \multirow{2}{1.8cm}{DF-OPT} &\multicolumn{5}{c}{LL-OPT} \\ \cline{4-8}
& & & $h$=1 & $h$=2 & $h$=3 & $h$=4 & $h$=5 \\ \hline
\multirow{2}{1.2cm}{$10^2$}
& 0.330 & 5.856 & 0.383 & 0.386 & 0.390 & 0.394 & 0.402 \\
& (0.004) & (5.114) & (0.619) & (0.621) & (0.624) & (0.628) & (0.634)  \\
\multirow{2}{1.2cm}{$10^3$}
& 0.511 & 1229.572 & 0.404 & 0.418 & 0.434 & 0.472 & 0.559 \\
& (0.008) & (876.831) & (0.635) & (0.647) & (0.659) & (0.687) & (0.748) \\
\multirow{2}{1.2cm}{$10^4$}
& 3.442 & * & 0.499 & 0.621 & 0.789 & 1.199 & 1.948 \\
& (0.039) & * & (0.707) & (0.788) & (0.889) & (1.095) & (1.396)  \\
\multirow{2}{1.2cm}{$10^5$}
& 47.955 & * & 1.885 & 3.133 & 4.946 & 9.149 & 16.354 \\
& (0.121) & * & (1.373) & (1.770) & (2.224) & (3.025) & (4.044) \\
\hline
\end{tabular}
\end{center}
\end{table}

\begin{table}[!htb]
\begin{center}
\caption{Estimation errors (in Hellinger distance) for Example~\ref{beta_mixture_example}. Average of 5 replicates, standard deviation reported in parentheses. The DF-OPT algorithm could not finish the two large-sample simulations within reasonable time. \label{table_semibeta_dis}}
\begin{tabular}{cccccccc} 
\hline
\multirow{2}{1.5cm}{Sample Size} & \multirow{2}{1.5cm}{OPT} & \multirow{2}{1.8cm}{DF-OPT} &\multicolumn{5}{c}{LL-OPT} \\ \cline{4-8}
& & & $h$=1 & $h$=2 & $h$=3 & $h$=4 & $h$=5 \\ \hline
\multirow{2}{1.2cm}{$10^2$}
& 0.381 & 0.381 & 0.451 & 0.384 & 0.392 & 0.380 & 0.388 \\
& (0.048) & (0.048) & (0.056) & (0.041) & (0.045) & (0.046) & (0.050) \\
\multirow{2}{1.2cm}{$10^3$}
& 0.183 & 0.183 & 0.205 & 0.183 & 0.185 & 0.180 & 0.183 \\
& (0.012) & (0.012) & (0.013) & (0.009) & (0.012) & (0.012) & (0.012) \\
\multirow{2}{1.2cm}{$10^4$}
& 0.081 & * & 0.102 & 0.081 & 0.082 & 0.082 & 0.081 \\
& (0.004) & * & (0.008) & (0.004) & (0.004) & (0.004) & (0.003) \\
\multirow{2}{1.2cm}{$10^5$}
& 0.034 & * & 0.045 & 0.035 & 0.034 & 0.034 & 0.034 \\
& (0.001) & * & (0.003) & (0.002) & (0.001) & (0.001) & (0.001) \\
\hline
\end{tabular}
\end{center}
\end{table}

Based on Table~\ref{table_semibeta_cost}, we find that the LL-OPT algorithm achieves substantial speedup compared with the exact OPT algorithm, especially when the sample size is large. The other algorithm for exact OPT inference, the DF-OPT algorithm, is computationally prohibitive for large sample size.

\remove{
Based on Table~\ref{table_semibeta_cost}, we examine our time complexity analysis in Section~\ref{complexity}. According to Theorem~\ref{thm_OPT}, the computational cost of the exact OPT algorithm is $O(pn^{\log_{1/t}2p})$, which is $O(n^{\log_{1/t}4})$ for $p=2$, while the data in Table~\ref{table_semibeta_cost} suggest that the computational complexity is roughly proportional to $N^{2.12}$, which corresponds to $t=0.52$. 2) According to Theorem~\ref{thm_LL-OPT}, for fixed $h$, $p$ and $r$, the computational cost of the LL-OPT algorithm is $f(n)=O\{(n+2^hp)p^h2^r\}=O(n)$, which matches the simulation results. For fixed $n$, $p$ and $r$, the time complexity can be written as $O(p^h)=O(2^h)$ since $n>2^hp$ in this case, which is consistent with Table~\ref{table_semibeta_cost} since the computational cost of the LL-OPT algorithm is roughly doubled when the value of $h$ is increased by 1.}

Based on Table~\ref{table_semibeta_dis}, the accuracy of estimates of the LL-OPT algorithm is comparable to that of the exact OPT algorithm, with the value of $h$ as low as 2, which is encouraging since the improvement on efficiency is sizable. For example, when $h=2$, the LL-OPT algorithm is $15$ times faster than the exact OPT algorithm for sample size of $10^5$. This makes the LL-OPT algorithm appealing for large samples. Furthermore, the estimates of the LL-OPT algorithms with $h \geqslant 2$ are almost identical, so the adaptive approach for selecting $h$ would work well. 

\begin{example}
\label{bivarirate_normal_example}
We now consider a distribution in $[0,1]^4$ with a bivariate Normal distribution in the first two dimensions and uniform in the other two dimensions. In this scenario, $X$ and $Y$ components are independent of each other and follow $Norm(0.6,0.1^2)$ and $Norm(0.4,0.1^2)$ respectively~\citep[adapted from Example 9]{WongLi2010}, and $Z$ and $W$ components are independent and uniform in $[0,1]$. As the probability mass of this distribution is essentially all located in $[0,1]\times [0,1]$, the density function can be approximated by
\begin{equation}
\label{bnorm} 
\frac{1}{2\pi \times 0.1^2} e^{-\frac{(x-0.6)^2+(y-0.4)^2}{2\times 0.1^2}}.
\end{equation}
\end{example}

\remove{Figure~\ref{bnorm1000} shows the partition maps based on the exact OPT algorithm and the LL-OPT algorithm with $h=1,\ldots, 6$ for a sample of size 1000. } The average running times and average Hellinger distances are summarized in Tables~\ref{table_bnorm_cost} and~\ref{table_bnorm_dis} respectively.  All the summary statistics are based on 5 replications and corresponding standard deviations are included in the parentheses.

\remove{
\begin{figure}[!htb]
\begin{center}
\includegraphics[scale=0.45]{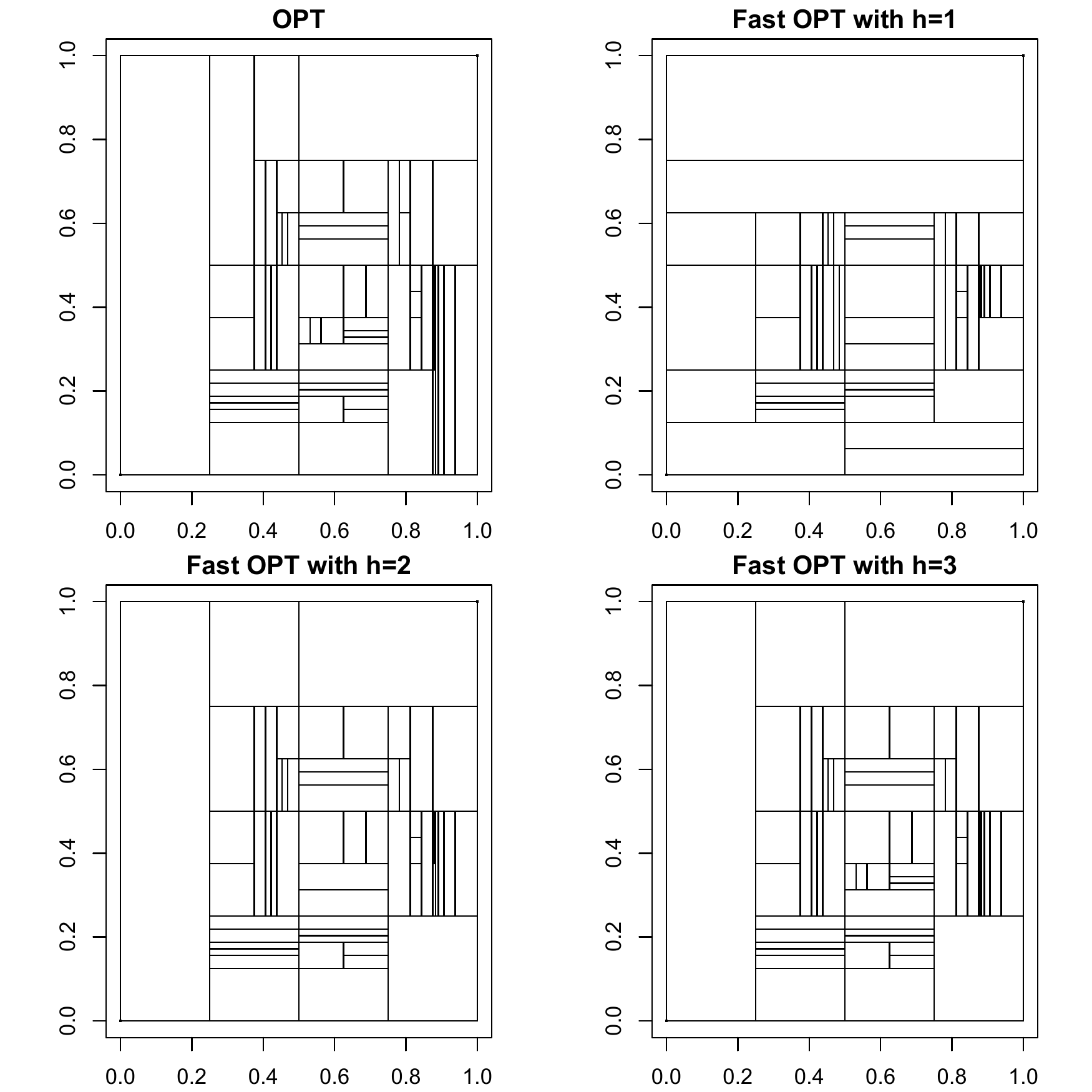}\includegraphics[scale=0.45]{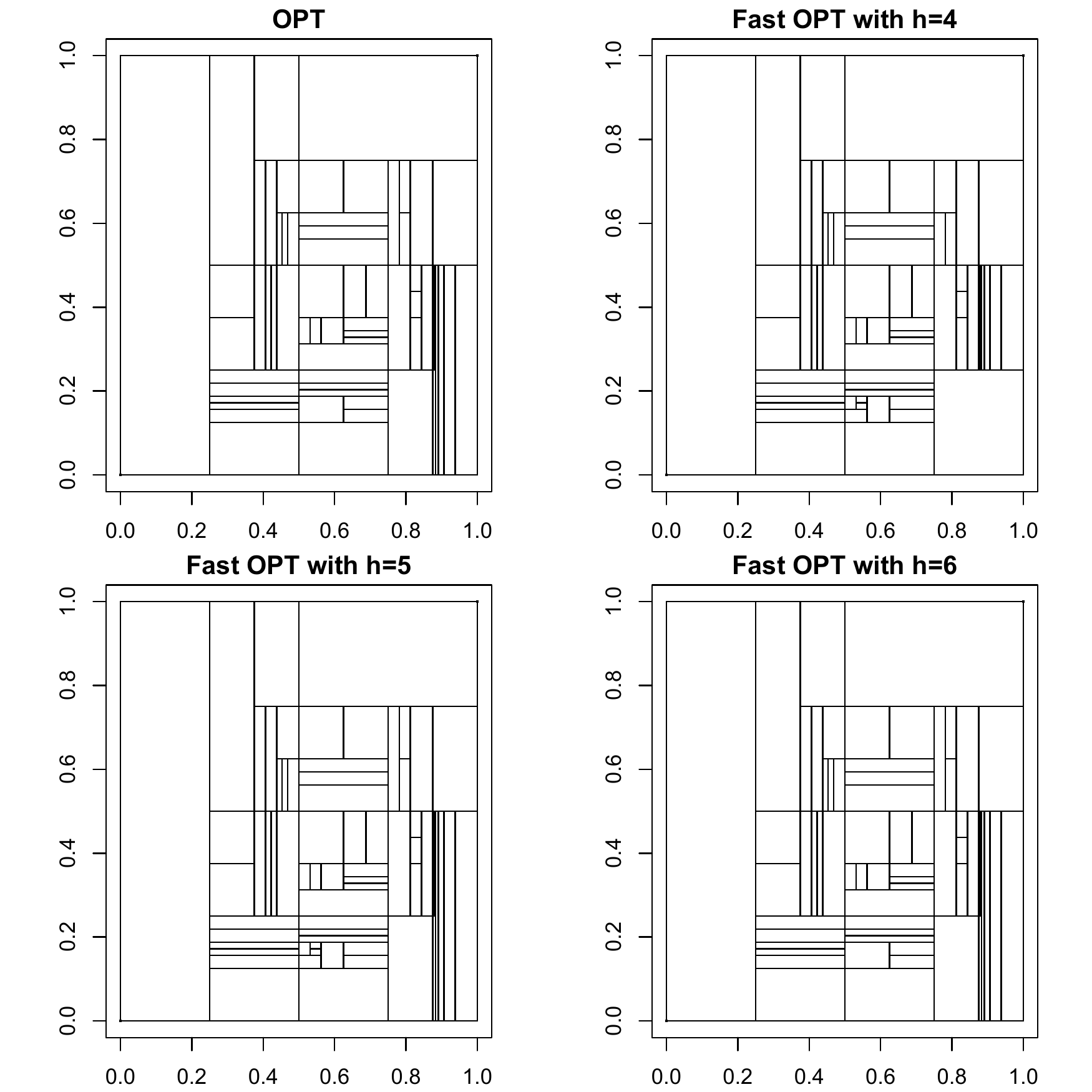}
\caption{The partition maps based on the exact OPT and the LL-OPT algorithms with $h=1,\ldots,6$. The samples are drawn from a bivariate Normal distribution as defined by~(\ref{bnorm}). Sample size is 1000.}
\label{bnorm1000}
\end{center}
\end{figure}
}

\remove{
\begin{table}[!htb]
\begin{center}
\begin{tabular}{| c | c | c | c| c | c | c |c|}
\hline
\multirow{2}{1.5cm}{Sample Size} & \multirow{2}{1.2cm}{OPT} 
&\multicolumn{6}{|c|}{LL-OPT} \\ \cline{3-8}
& & h=1 & h=2 & h=3 & h=4 & h=5& h=6 \\ \hline
\multirow{2}{1.2cm}{N=100} 
& 2.98 & 0.03& 0.11 & 0.25 & 0.51 & 0.95 & 1.56 \\ 
& (0.33) & (0.01) & (0.02) & (0.05) & (0.10) & (0.16)  & (0.21)\\  \hline
\multirow{2}{1.2cm}{N=250} 
& 17.53 & 0.05 & 0.19 & 0.47 & 1.06 & 2.19 & 4.25 \\ 
& (0.83) & (0.01) & (0.04) & (0.10) & (0.24) & (0.48)  & (0.85)\\  \hline
\multirow{2}{1.2cm}{N=500} 
& 79.52 & 0.08 & 0.30 & 0.86 & 2.19 & 5.28 & 11.56 \\ 
& (3.98) & (0.01) & (0.05) & (0.16) & (0.44) & (1.17)  & (2.58)\\  \hline
\multirow{2}{1.2cm}{N=1000} 
& 363.15 & 0.11& 0.44 & 1.36 & 3.93 & 11.14 & 28.82 \\ 
& (15.46) & (0.02) & (0.06) & (0.20) & (0.57) & (1.84)  & (4.86)\\  \hline
\end{tabular}
\end{center}
  \caption{The CPU time (in seconds) of executing different OPT algorithms. The samples are drawn from a bivariate Normal distribution as defined by~(\ref{bnorm})}
\label{table_bnorm_cost}
\end{table}
}

\begin{table}[!htb]
\begin{center}
\caption{Running times (in seconds) for Example~\ref{bivarirate_normal_example}. Average of 5 replicates, standard deviation reported in parentheses. The exact OPT algorithm could not finish the simulation with  $10^5$ samples in under 128GB of memory. \label{table_bnorm_cost}}
\begin{tabular}{cccccccc} 
\hline
\multirow{2}{1.5cm}{Sample Size} & \multirow{2}{1.5cm}{OPT}  &\multicolumn{6}{c}{LL-OPT} \\ \cline{3-8}
& & $h$=1 & $h$=2 & $h$=3 & $h$=4 & $h$=5 & $h$=6 \\ \hline
\multirow{2}{1.2cm}{$10^2$}
& 1.411 & 0.395 & 0.388 & 0.426 & 0.591 & 1.076 & 2.960 \\
& (0.074) & (0.033) & (0.002) & (0.010) & (0.010) & (0.028) & (0.061) \\
\multirow{2}{1.2cm}{$10^3$}
& 31.160 & 0.405 & 0.535 & 0.768 & 1.911 & 6.490 & 25.609 \\
& (0.532) & (0.018) & (0.048) & (0.035) & (0.035) & (0.129) & (0.432) \\
\multirow{2}{1.2cm}{$10^4$}
& 778.208 & 0.561 & 1.472 & 4.072 & 14.475 & 58.316 & 248.463 \\
& (3.028) & (0.022) & (0.099) & (0.199) & (0.197) & (0.589) & (2.529) \\
\multirow{2}{1.2cm}{$10^5$}
& * & 2.904 & 12.654 & 40.836 & 154.157 & 615.615 & 2638.182 \\
& * & (0.055) & (0.585) & (0.301) & (1.404) & (0.995) & (52.133) \\
\hline
\end{tabular}
\end{center}
\end{table}

\remove{
\begin{table}[!htb]
\begin{center}
\begin{tabular}{| c | c | c | c| c | c | c |c|} 
\hline
\multirow{2}{1.5cm}{Sample Size} & \multirow{2}{1.2cm}{OPT} 
&\multicolumn{6}{|c|}{LL-OPT} \\ \cline{3-8}
& & h=1 & h=2 & h=3 & h=4 & h=5& h=6 \\ \hline
\multirow{2}{1.2cm}{N=100} 
& 0.0753 & 0.0840 & 0.0744 & 0.0751 & 0.0753 & 0.0757 & 0.0754 \\ 
& (0.0119) & (0.0156) & (0.0122) & (0.0121) & (0.0120) & (0.0117)  & (0.0117)\\  \hline
\multirow{2}{1.2cm}{N=250} 
& 0.0488 & 0.0515 & 0.0471 & 0.0486 & 0.0490 & 0.0489 & 0.0490 \\ 
& (0.0059) & (0.0084) & (0.0052) & (0.0058) & (0.0060) & (0.0059)  & (0.0059)\\  \hline
\multirow{2}{1.2cm}{N=500} 
& 0.0356 & 0.0372 & 0.0345 & 0.0353 & 0.0354 & 0.0355 & 0.0356 \\ 
& (0.0039) & (0.0053) & (0.0039) & (0.0039) & (0.0039) & (0.0040)  & (0.0040)\\  \hline
\multirow{2}{1.2cm}{N=1000} 
& 0.0257 & 0.0261 & 0.0249 & 0.0256 & 0.0256 & 0.0257 & 0.0257 \\ 
& (0.0022) & (0.0024) & (0.0020) & (0.0020) & (0.0021) & (0.0021)  & (0.0022)\\  \hline
\end{tabular}
\end{center}
  \caption{The Hellinger distances between the estimated and true densities. The samples are drawn from a bivariate Normal distribution as defined by~(\ref{bnorm}).}
\label{table_bnorm_dis}
\end{table}
}

\begin{table}[!htb]
\begin{center}
\caption{Estimation errors (in Hellinger distance) for Example~\ref{bivarirate_normal_example}. Average of 5 replicates, standard deviation reported in parentheses. The exact OPT algorithm could not finish the simulation with $10^5$ samples in under 128GB of memory.\label{table_bnorm_dis}}
\begin{tabular}{cccccccc} 
\hline
\multirow{2}{1.5cm}{Sample Size} & \multirow{2}{1.5cm}{OPT} &\multicolumn{6}{c}{LL-OPT} \\ \cline{3-8}
& & $h$=1 & $h$=2 & $h$=3 & $h$=4 & $h$=5 & $h$=6 \\ \hline
\multirow{2}{1.2cm}{$10^2$}
& 0.561 & 0.922 & 0.612 & 0.571 & 0.564 & 0.562 & 0.561 \\
& (0.019) & (0.005) & (0.029) & (0.029) & (0.021) & (0.022) & (0.018) \\
\multirow{2}{1.2cm}{$10^3$}
& 0.383 & 0.917 & 0.391 & 0.380 & 0.383 & 0.382 & 0.382 \\
& (0.004) & (0.005) & (0.005) & (0.003) & (0.004) & (0.004) & (0.004) \\
\multirow{2}{1.2cm}{$10^4$}
& 0.258 & 0.914 & 0.253 & 0.256 & 0.258 & 0.258 & 0.258 \\
& (0.001) & (0.005) & (0.001) & (0.002) & (0.003) & (0.002) & (0.001) \\
\multirow{2}{1.2cm}{$10^5$}
& * & 0.916 & 0.168 & 0.171 & 0.173 & 0.173 & 0.173 \\
& * & (0.001) & (0.001) & (0.001) & (0.001) & (0.001) & (0.000) \\
\hline
\end{tabular}
\end{center}
\end{table}

Similar to the previous example, the estimation accuracy of the LL-OPT algorithm approaches that of the exact OPT algorithm as the value of $h$ increases. Patterns of running times found in Table~\ref{table_bnorm_cost} are also similar to that of the previous example. According to Table~\ref{table_bnorm_dis}, the estimation accuracy of the LL-OPT algorithm with $h$ as low as $3$ is comparable to that of the exact OPT algorithm. The improvements on running times are still sizable. When $h=3$, the LL-OPT algorithm is over $40$ times faster than the exact OPT algorithm for sample size of $10^3$, and over $190$ times faster for sample size of $10^4$. Finally, estimation based on the LL-OPT algorithm stabilizes for $h\geq3$.

In both simulation examples, when the sample size is large (e.g., $10^5$), we can see that the running time of the LL-OPT algorithm increases roughly by a factor of $p$ when $h$ increases by 1, which is consistent with the time complexity analysis provided in Theorem~\ref{thm_LL-OPT}.

Simulations were run on an Intel Xeon X7560 2.27Ghz on a single core.

\subsection{Smooth density estimation}

We compare the performance of the exact OPT, the LL-OPT and the FEE algorithms for density estimation to the kernel density estimation (KDE)~\citep{Silverman1986} method. 
\remove{Hellinger distance is a commonly used metric in Bayesian non-parametrics. It is the $L_2$ distance of $\sqrt{f(x)}$ and $\sqrt{g(x)}$; the asymptotic mean integrated squared error which is abundant in the kernel estimation literature is not a metric on $\mathcal{P}_1(\Omega)$, so it does not fit our purpose because it is not convenient for us to analyze the properties of FEE on the whole functional space.}

\begin{example}
We consider a strongly skewed distribution
\begin{equation}
%X\sim\sum_{i=0}^7N\left(3\left(\left(\frac23\right)^i - 1\right),\left(\frac23\right)^{2i}\right).
x\sim\Gamma(2, 0.1)\mathbf{I}(0, 1)
\label{1d_eq}
\end{equation}
where $\Gamma(a, b) = \frac{1}{b^a\Gamma(a)}x^{a - 1}e^{\frac{-x}{b}}$.
\end{example}

Figure~\ref{1d_fig} shows the estimated densities by the exact OPT, the LL-OPT and the FEE algorithms for sample size $N=10^4$. Table~\ref{1d_tbl} shows the estimation errors measured by Hellinger distance. FEE gives the best results among all the methods.

\begin{figure}[!htb]
\center
\includegraphics[scale=.5]{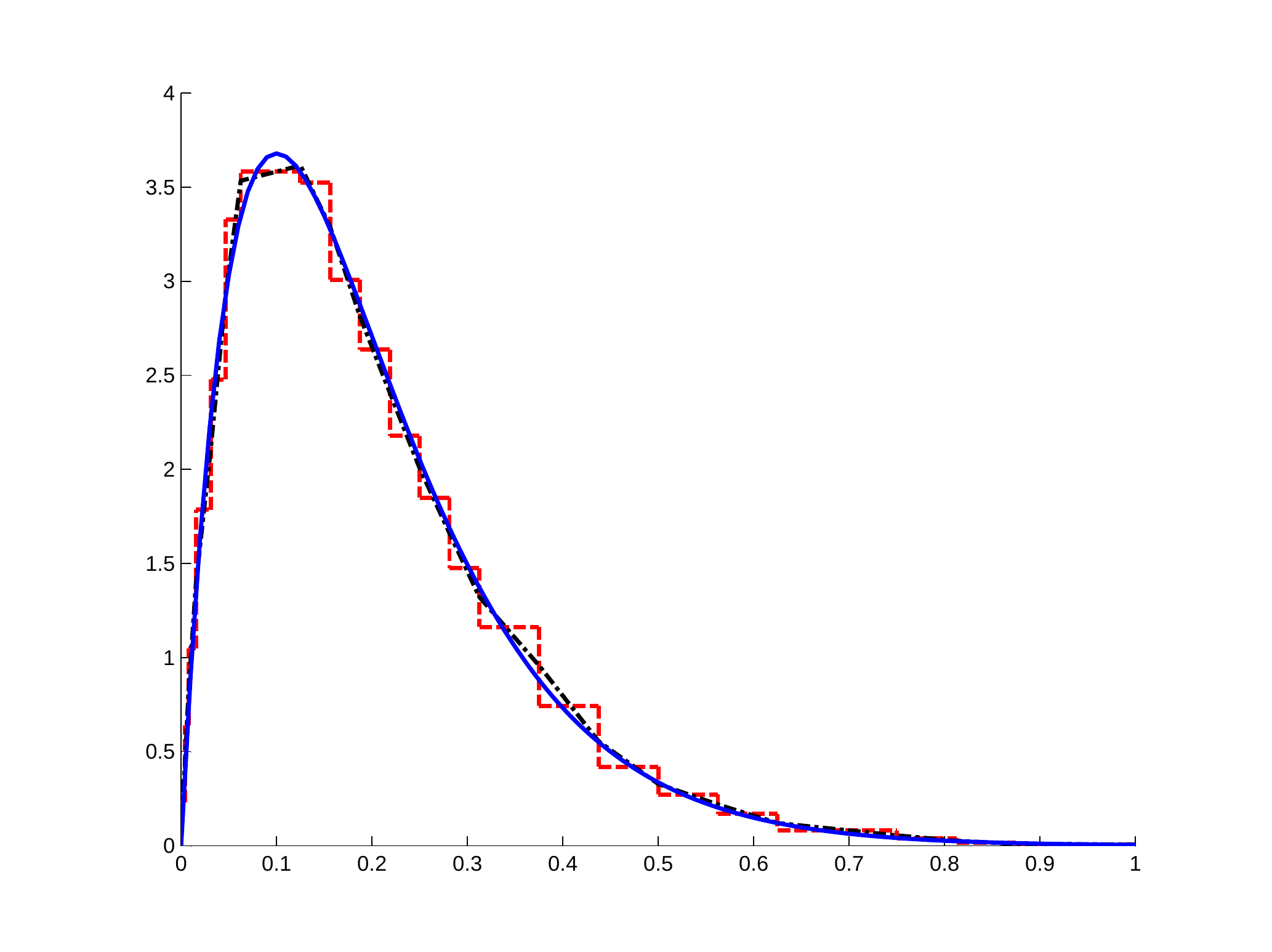}
\caption{Plot of densities estimated with $10^4$ samples simulated from~(\ref{1d_eq}) by OPT and LL-OPT in red (they are complete overlapped), and FEE in black. The true density is in blue. \label{1d_fig}}
\end{figure}

\begin{table}[!htb]
\center
\caption{Estimation errors (in Hellinger distance) of OPT, LL-OPT, FEE and KDE for density~(\ref{1d_eq}). Average of 10 replicates, standard deviation reported in parentheses.  \label{1d_tbl}}
\begin{tabular}{ccccc}
\hline
Sample Size&OPT&LL-OPT&FEE&KDE\\
\hline
\multirow{2}{1.2cm}{$10^2$}&$0.1749$&$0.1674$&$0.1341$&$0.0984$\\
&$(0.0323)$&$(0.0221)$&$(0.0223)$&$(0.0180)$\\
\multirow{2}{1.2cm}{$10^3$}&$0.0772$&$0.0753$&$0.0397$&$0.0474$\\
&$(0.0068)$&$(0.0105)$&$(0.0061)$&$(0.0036)$\\
\multirow{2}{1.2cm}{$10^4$}&$0.0385$&$0.0385$&$0.0192$&$0.0297$\\
&$(0.0020)$&$(0.0043)$&$(0.0020)$&$(0.0012)$\\
%$10^5$&&&&\\
\hline
\end{tabular}
\end{table}

\begin{example}
\label{2d_example}
We consider a mixture of a uniform distribution and Beta distributions.
\begin{equation}
X_1\sim U(0, 1);~X_2\sim\frac{4}{5}\beta(2, 10) + \frac{1}{5}\beta(7, 2)
\label{2d_eq}
\end{equation}
\end{example}

Figure~\ref{2d_fig}~(a) shows the densities estimated by the LL-OPT and the FEE algorithms. As a comparison, kernel density estimated is in Figure~\ref{2d_fig}~(b). We see that the uniformity in the first dimension of the distribution is well captured by OPT based approaches, but not by KDE.

\begin{figure}[!htb]
\begin{center}
\begin{tabular}{cc}
\includegraphics[scale=.5]{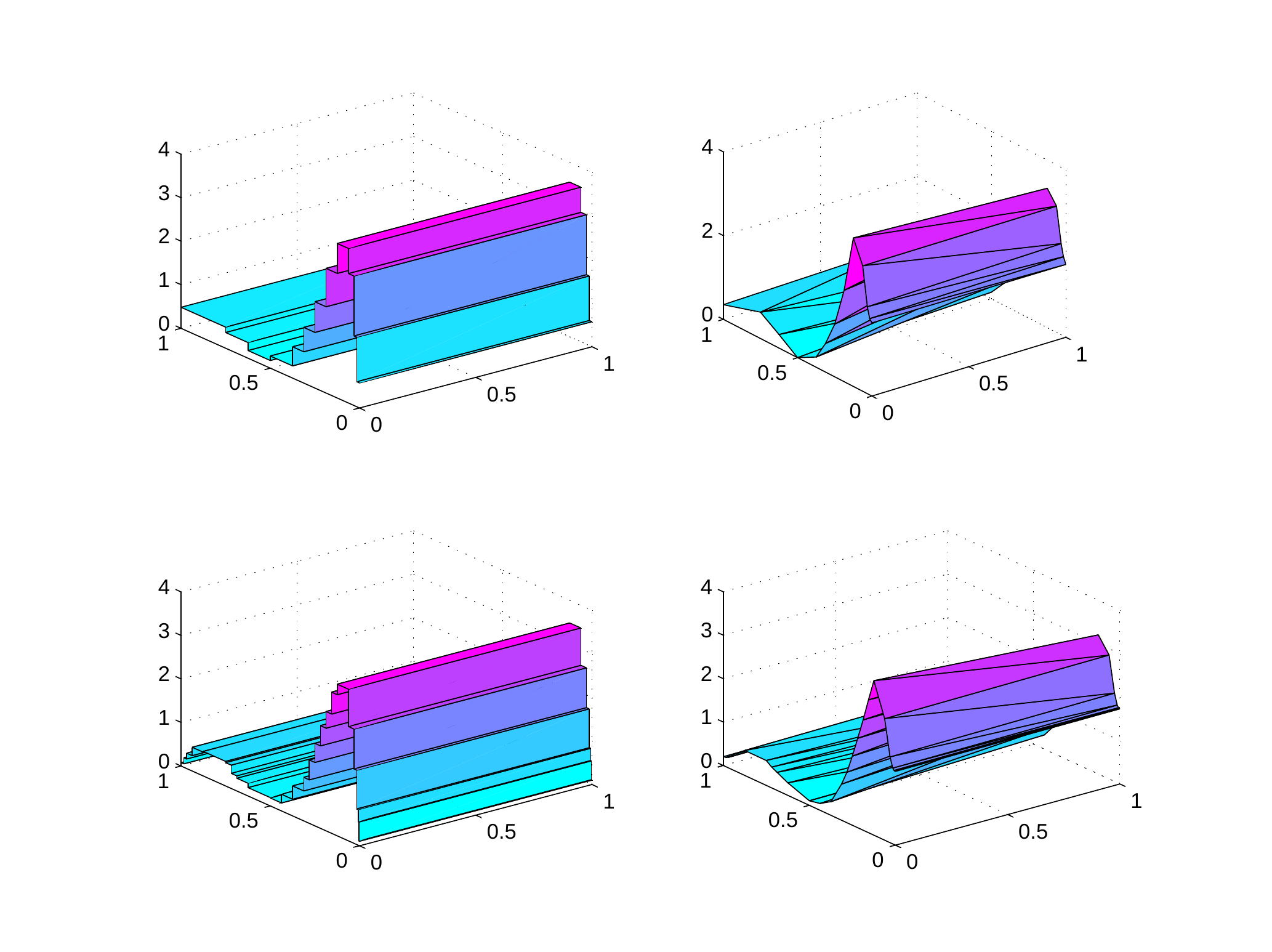} & \hspace{-1cm}\includegraphics[scale=.4]{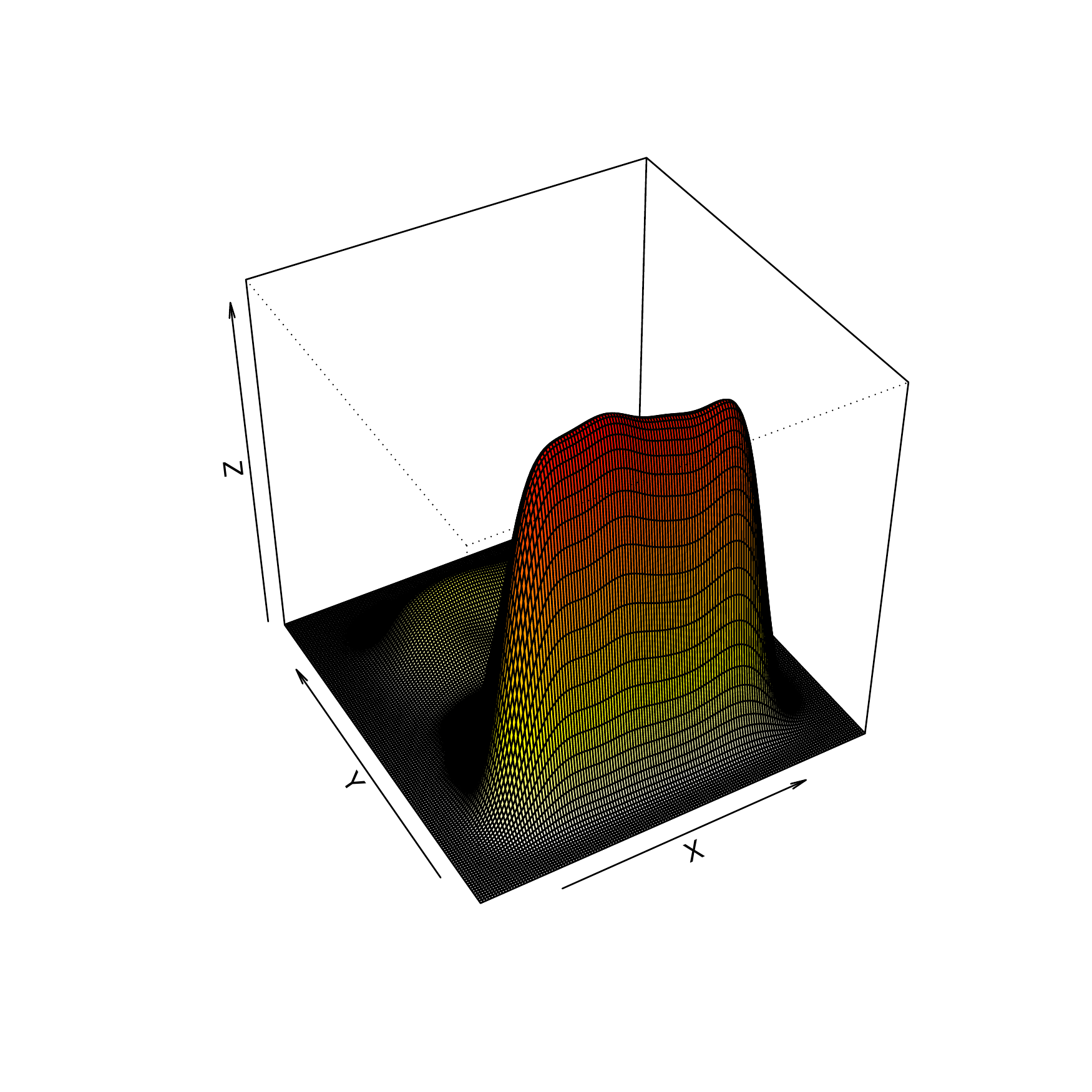}\\
(a)&(b)\\
\end{tabular}
\caption{The estimated densities for Example~\ref{2d_example}. (a) Densities estimated by LL-OPT (left) and FEE (right) with $10^3$ and $10^4$ samples simulated from~(\ref{2d_eq}) respectively. (b) Density estimated by kernel density estimation with $10^4$ samples. \label{2d_fig}}
\end{center}
\end{figure}

\begin{example}
\label{example_high_dim}
We consider a 5-dimensional random vector $$\mathbf{Y} = (X_1, \frac{1}{3}X_1 + \frac{2}{3}X_2, X_3, X_4, \frac{1}{5}X_3 + \frac{4}{5}X_5),$$ where $X_1,\ldots,X_5$ are independent and sampled from
\begin{eqnarray}
  X_1 \sim \beta(2, 8), X_2\sim\beta(8,2), X_3\sim U(0, 1)\\
  X_4 \sim \Gamma(2, 1)\mathbf{I}(0, 1), X_5\sim \Gamma(1, 2)\mathbf{I}(0, 1)
\end{eqnarray}
\end{example}

With a sample size of $6\times10^3$, the Hellinger distances for the LL-OPT and the FEE algorithms are 0.31 and 0.27, respectively. Figure~\ref{highdim} shows marginal histograms based on $10^4$ random samples drawn from the fitted FEE density and the true density. We can see that the fitted density resembles closely the true density in every dimension.

\begin{figure}[!htb]
\begin{center}
\includegraphics[scale=.5]{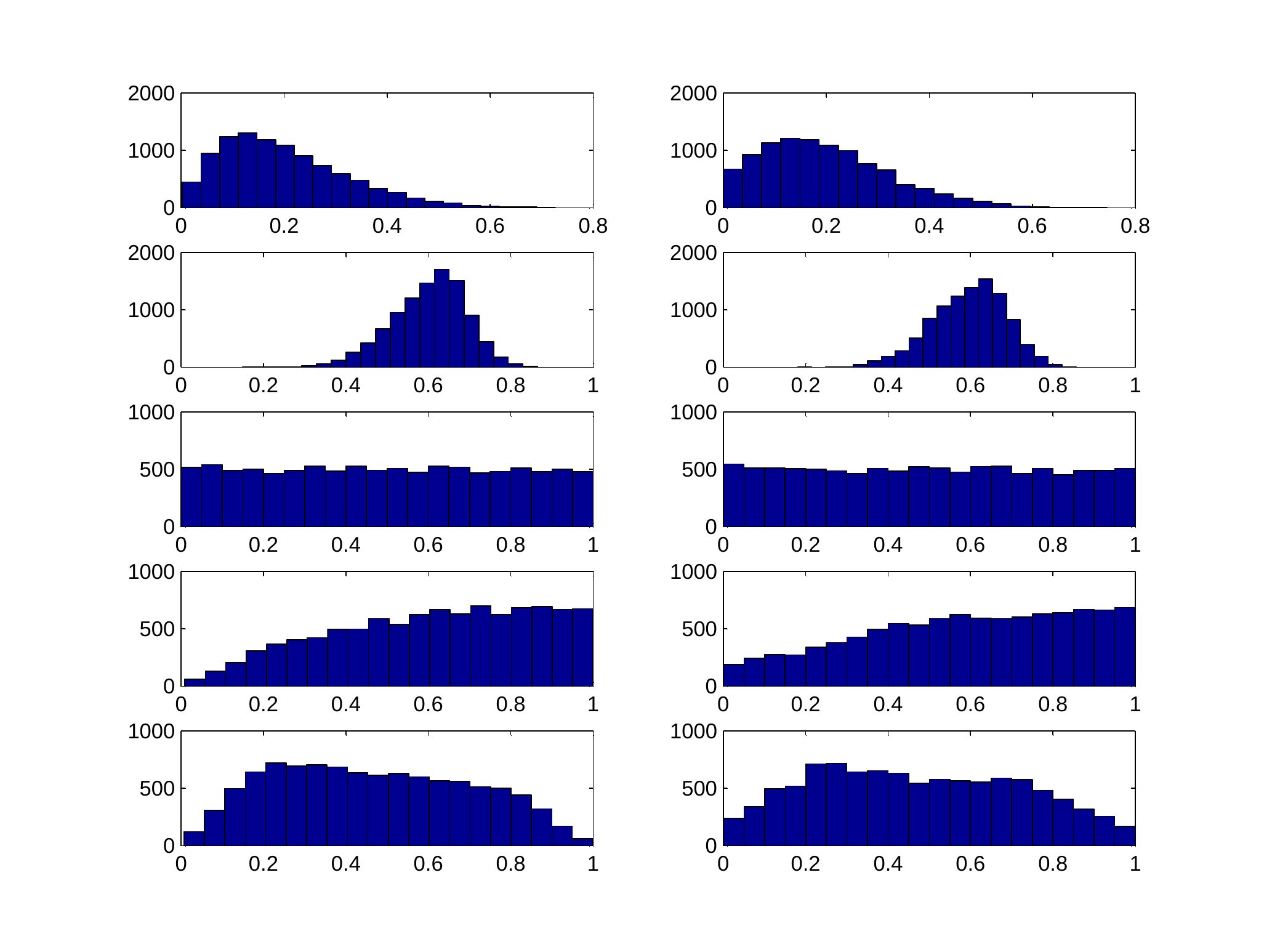}
\caption{The marginal histograms of $10^4$ random samples drawn from the fitted FEE density and the true density for Example~\ref{example_high_dim}. From top to bottom are the histograms of $Y_1$ to  $Y_5$.\label{highdim}}
\end{center}
\end{figure}

\section{Discussion}

We provide two software packages implementing the algorithms described in this paper, available at \url{http://www.stanford.edu/group/wonglab/opt_comp/}. The \textbf{fast-opt} package implements both the exact OPT and the LL-OPT algorithms described in Section 2 and 3. The \textbf{smooth-opt} package implements the smoothing algorithm described in Section 5. Both packages are implemented in C++.

The \textbf{smooth-opt} package relies on the user to specify the value for the tuning parameter $\lambda$. An alternative approach is to use cross-validation to choose $\lambda$ which requires much more computation. In our simulations we found that the result is not sensitive to the value of $\lambda$ and we use $\lambda=10^{-3}$ or $10^{-4}$ in all our simulations.

High dimensional density estimation is a challenging task.\remove{significantly different from univariate density estimation. For instance, it is very difficult to visualize the estimated density when the dimension of the sample space is more than two. Although kernel density estimators can be generalized to high dimensional cases~\citep{Silverman1986}, the performance of these estimators suffer from the infamous curse of dimensionality. High dimensional density estimation therefore has diminishing usefulness and serves mainly as an exploratory data analysis tool, where visualization is not even possible.} The OPT is an adaptive non-parametric density estimation approach which fits to the global landscape of the data adaptively rather than locally as in kernel density estimation~\citep{Silverman1986}, as demonstrated in our simulations. The simulations also suggest that our continuous piecewise linear FEE further improves upon the piecewise constant counterpart.\remove{Even in relatively high dimensional cases (dimension~$=5$), the results are still reasonably good.} The proposed FEE is a penalized regression approach. Based on a fast quadratic programming solver, it achieves decent efficiency when the sample size is reasonably high ($10^5$ data points). The resulting estimated density function is continuous piecewise linear, which is mostly sufficient when we want to explore how the probability mass is distributed or visualize the estimated density. When necessary, the FEE can be extended to use higher order basis functions, at the cost of heavier computation.

%\section*{Appendix}

\section*{Acknowledgements}
This work was supported by NSF grant DMS-09-06044 and NIH grant R01-HG006018. We thank IBM$\textsuperscript{\textregistered}$ for providing the academic license of Cplex$\textsuperscript{\textregistered}$.

\bibliographystyle{abbrvnat}
\bibliography{opt_computation}

\end{document}